\documentclass[12pt]{article}
\usepackage{setspace}
\usepackage{physics}
\usepackage{dcolumn} 
\usepackage{bm}     
\usepackage{braket}
\usepackage{amsmath,amssymb,amsfonts,amsthm}
\usepackage{mathtools}
\usepackage{xcolor}
\usepackage{tikz}
\usepackage{xparse}
\usepackage{hyperref}
\usepackage{cleveref}
\usepackage{listings}
\usepackage{comment}
\usepackage{dsfont}
\usepackage[T1]{fontenc}
\usepackage{url}
\usepackage{stmaryrd} % knuth notation
\usepackage{graphicx}
\usepackage{hyperref}
\usepackage[normalem]{ulem}
\usepackage{nameref}
\usepackage{authblk}
\hypersetup{colorlinks=true,citecolor=blue,linkcolor=blue,filecolor=blue,urlcolor=blue}

\providecommand{\ignore}[1]{}

\newif\ifcmnt
\cmnttrue
\ifdefined\cmntsoff\cmntfalse\fi

\ifcmnt

\newtheorem{thm}{Theorem}[section]
\newtheorem{prop}[thm]{Proposition}
\newtheorem{lem}[thm]{Lemma}

\newtheorem{cor}[thm]{Corollary}

\theoremstyle{definition}
\newtheorem{definition}[thm]{Definition}

\newtheorem{example}[thm]{Example}
\numberwithin{equation}{section}

\numberwithin{equation}{section}

\newtheorem*{theorem*}{Theorem}

\newtheorem*{proposition*}{Proposition}

\newtheorem*{lemma*}{Lemma}

\newcommand{\al}{\alpha}
\newcommand{\la}{\lambda}
\newcommand{\one}{\mathtt{1}}
\newcommand{\Ups}{\Upsilon}
\newcommand{\da}{\downarrow}
\newcommand{\ua}{\uparrow}
\newcommand{\varnot}{\varnothing}
\newcommand{\ga}{\gamma}
\newcommand{\Og}{\Omega}

\newcommand{\cM}{\mathcal{M}}
\newcommand{\cN}{\mathcal{N}}

\newcommand{\cP}{\mathcal{P}}
\newcommand{\cQ}{\mathcal{Q}}

\newcommand{\cT}{\mathcal{T}}

\newcommand{\bC}{\mathbb{C}}
\newcommand{\bR}{\mathbb{R}}

\begin{document}
\title{Refining Ky Fan's majorization relation with linear programming}

\author{Mohammad A. Alhejji\footnote{Correspondence: malhejji@unm.edu}}
\affil{Center for Quantum Information and Control, University of New Mexico, Albuquerque, NM 87131, USA.}

\maketitle
\begin{abstract}
A separable version of Ky Fan's majorization relation is proven for a sum of two operators that are each a tensor product of two positive semi-definite operators. In order to prove it, upper bounds are established on the relevant largest eigenvalue sums in terms of the optimal values of certain linear programs. The objective function of these linear programs is the dual of the direct sum of the spectra of the summands. The feasible sets are bounded polyhedra determined by positive numbers, called alignment terms, that quantify the overlaps between pairs of largest eigenvalue spaces of the summands. By appealing to geometric considerations, tight upper bounds are established on the alignment terms of tensor products of positive semi-definite operators. As an application, the spin alignment conjecture in quantum information theory is affirmatively resolved to the 2-letter level. Consequently, the coherent information of platypus channels is additive to the 2-letter level. 
\end{abstract}

\begin{comment}
So, you're looking at my source code. Why?
\end{comment}

\section{Introduction}
\label{sec: intro}
Among the open problems in quantum information theory, the problem of regularized formulas sticks out like a sore thumb. Put briefly, known formulas for the optimal rates of asymptotically error-free information processing---a principal concern of the theory--- are given by limits of sequences of optimizations over domains with exponentially growing dimensions. Often, little is known about the convergence rates of these sequences. It is commonly surmised that much is yet to be understood about information processing in quantum-mechanical contexts. To make progress, new analysis tools are necessary.

Let us consider the transmission of quantum information through a memoryless quantum channel \(\cN\) (a linear completely positive trace-preserving map). The quantum capacity of \(\cN\) is the largest rate of asymptotically error-free transmission of quantum information through \(\cN\), measured in qubits per channel use. Let \(\cQ (\cN)\) denote the quantum capacity of \(\cN\). We have one formula for \(\cQ  (\cN)\), due to Lloyd \cite{Lloyd1997}, Shor \cite{shor2002} and Devetak \cite{Devetak2005}, in terms of the coherent information of \(\cN\). Let \(\cN^c\) denote a complementary channel to \(\cN\) defined via a Stinespring dilation (see, for example, Ch.~5 of Ref.~\cite{Wilde2013}). The output of \(\cN^c\) is conventionally called the environment of \(\cN\). The coherent information of \(\cN\) is given by  
\begin{align}
\label{eq: coherent info}
I^{\text{(coh)}} (\cN) := \max_{\rho} H( \cN (\rho) ) - H (\cN^c (\rho)),
\end{align}
where \(H (X) := -\tr (X \log_2 X)\) is the von Neumann entropy. The maximization is over quantum states (positive semi-definite operators with unit trace) in the domain of \(\cN\). The formula we have for \(\cQ (\cN)\) is the regularized formula:
\begin{align}
\label{eq: quantum capacity}
\cQ (\cN) = \lim_{n \rightarrow \infty} \frac{I^{\text{(coh)}} (\cN^{\otimes n})}{n}.
\end{align}

It follows from the additivity of the von Neumann entropy under tensor products that \(I^{\text{(coh)}}(\cN^{\otimes n}) \geq n I^{\text{(coh)}} (\cN)\). If \(I^{\text{(coh)}}(\cN^{\otimes n}) = n I^{\text{(coh)}} (\cN)\), then \(\cN\) is said to have \textit{additive} coherent information to the \(n\)-letter level. If this is the case for all \(n\), then \(\cQ (\cN)\) has the manifestly computable formula \(I^{\text{(coh)}} (\cN)\). Unfortunately, this is generically not the case (see, for example,  Sec.~IV.B of Ref.~\cite{Koudia2022}). A stark instance of this \textit{non-additivity} is found in Ref.~\cite{Cubitt2015}, where Cubitt et al. show that for every positive integer \(m\), there exists a quantum channel \(\cN_m\) such that \(I^{\text{(coh)}} (\cN_m^{\otimes m}) = 0\) while \(\cQ(\cN_m) > 0\).

Here we are concerned with a conjecture called the spin alignment conjecture. Introduced by Leditzky et al. in Ref.~\cite{Leditzky_2023}, it is a conjecture about the entropies of convex mixtures of quantum states that have a certain tensor product structure. It arose in the context of computing the quantum capacities of platypus channels. The particular form of these quantum channels is not relevant to us here. By a reduction in Sec.~V of Ref.~\cite{Leditzky_2023}, each platypus channel \(\cM\) has an associated quantum state \(M\) such that the coherent information of \(\cM\) is additive to the \(n\)-letter level if for all probability distributions \(p\) on the subsets of \([n]:=\{1,2,\ldots, n\}\) and all tuples of pure (rank-\(1\)) quantum states \((\psi_J)_{J \subseteq [n]}\), it holds that
\begin{align}
\label{eq: align entropy}
H(\sum_{J \subseteq [n]} p_{J}\psi_{J} \otimes M^{\otimes J^c}) \geq H(\sum_{J \subseteq [n]} p_{J}(v_M v_M^*)^{\otimes J} \otimes M^{\otimes J^c}), 
\end{align}
where \(v_M\) is a unit eigenvector of \(M\) corresponding to its largest eigenvalue. For each subset \(J \subseteq[n]\), the expression \(\psi_J \otimes M^{\otimes J^c}\) denotes a quantum state that is separable across the bipartition \(J \mid J^c\). The factor \(\psi_J\) denotes the quantum state of particles whose labels are in \(J\) and the tensor power \(M^{\otimes J^c}\) denotes the quantum state of particles whose labels are in the complement \(J^c\). The spin alignment conjecture is that the inequality above holds in general. Results concerning the interplay between platypus channels and other quantum channels are shown in Ref.~\cite{Leditzky2022b}, some of which are conditioned on this conjecture.

Described in words, the entropy of the convex mixture \(\sum_{J \subseteq[n]} p_{J}\psi_{J} \otimes M^{\otimes J^c}\), which may be interpreted as the quantum state of the environment of \(\cM^{\otimes n}\), is conjectured to be smallest when \((\psi_J)_{J \subseteq[n]}\) is chosen so that the particles have maximally aligned spins. The stronger conjecture that 
\begin{align}
\label{rel: spin alignment maj}
\sum_{J \subseteq [n]} p_{J}\psi_{J} \otimes M^{\otimes J^c} \preceq \sum_{J \subseteq [n]} p_{J}({v_M} {v_M}^*)^{\otimes J}  \otimes M^{\otimes J^c},
\end{align}
where \(\preceq\) is the majorization symbol (see Sec.~\ref{sec: prelims} for a formal definition), was pursued in Ref.~\cite{Alhejji2024}. As of this writing, a complete resolution of either conjecture is yet to be found.

It is instructive to consider the \(n=1\) case.  For a self-adjoint (or Hermitian) operator \(A\) on \(\bC^d\), \(\la (A) \in \bR^d\)  denotes the tuple of eigenvalues of \(A\) arranged so \(\la_1 (A) \geq \la_2 (A) \geq \cdots \geq \la_d (A)\). The relation above immediately follows from Ky Fan's majorization relation~\cite{Fan1949} (also known as Ky Fan's eigenvalue inequality~\cite{Moslehian_2012}):
\begin{align}
\la(p_{\{1\}} \psi_{\{1\}} + p_{\varnot} M) &\preceq \la(p_{\{1\}} \psi_{\{1\}}) + \la(p_{\varnot} M)\\
&= \la(p_{\{1\}} v_M v_M^*) + \la(p_{\varnot} 
M)\\
&= \la(p_{\{1\}} v_M v_M^* + p_{\varnot} M).
\end{align}
The last equality is crucial to this argument. More generally, for self-adjoint operators \(A_1\) and \(A_2\), the equality \(\la(A_1) + \la(A_2) = \la(A_1 + A_2)\) holds if and only if there exist orthonormal vectors \(w_1, \ldots, w_{d}\) such that \(A_1 w_i = \la_i (A_1) w_i\) and \(A_2 w_i = \la_i (A_2) w_i\) for each \(i \in [d]\) (see, for example, Ref.~\cite{Niezgoda2022}). This invites the following definitions. 
\begin{definition}
\label{def: alignment map}
Let \(q_1, \ldots, q_d\) denote elements of the standard orthonormal basis for \(\bC^d\). We define an \textit{alignment} map \(^\da\) on the space of self-adjoint operator on \(\bC^d\) by the action \(X \mapsto X^\da :=  \sum_{i=1}^d \la_i (X) q_i q_i^*\).
\end{definition}
\begin{definition}
\label{def: perfect alignment}
Given self-adjoint operators \(X_1, X_2, \ldots, X_m\) on \(\bC^d\), we say that they are \textit{perfectly aligned} if there exists a unitary operator \(U\) on \(\bC^d\) such that \(X_i^{\da} = U X_i {U}^*\) for each \(i \in [m]\).
\end{definition}

By considering majorization as a preorder on the space of self-adjoint operators, Fan's majorization relation may be written as follows:
\begin{align}
\label{eq: Ky Fan majorization relation with da}
A_1 + A_2 \preceq A_1^\da + A_2^\da.
\end{align}
This relation generalizes to conic combinations of self-adjoint operators.

The obstacle to extending the argument above to cases where \(n > 1\) is that, in general, there may not exist a pure quantum state tuple \((\psi_J)_{J \subseteq[n]}\) such that the summands \(p_{J}\psi_J \otimes M^{\otimes J^c}\) are perfectly aligned. However, at the conjectured optimal tuple \(({v_M {v_M}^*}^{\otimes J})_{J \subseteq[n]}\), the quantum state of each particle is a convex combination of the perfectly aligned operators \(M\) and \(v_M v_M^*\). This suggests that local perfect alignment may yield minimum entropy when global perfect alignment is impossible.

The main contribution of this work is a formalization of this suggestion for a sum of two tensor products of two positive semi-definite operators. 
\begin{thm}
\label{thm: sep Fan majorization}
Let \(B_1, B_2\) denote two positive semi-definite operators on \(\bC^{d_B}\) and \(C_1, C_2\) denote two positive semi-definite operators on \(\bC^{d_C}\). Then 
\begin{align}
\label{eq: sep Fan majorization relations}
B_1 \otimes C_1 + B_2 \otimes C_2 \preceq B_1^{\da} \otimes C_1^{\da} + B_2^{\da} \otimes C_2^{\da}.
\end{align}
\end{thm}

We call this relation a \textit{separable Fan majorization relation}. We use it to resolve the spin alignment conjecture to the \(2\)-letter level (Prop.~\ref{prop: spin alignment n=2}). We prove it by establishing linear programming refinements of Ky Fan's majorization relation (Thm.~\ref{thm: linear programming bound}). A crucial ingredient of the proof is an observation about the geometry of subspaces that are spanned by tensor product vectors (Prop.~\ref{prop: orthogonal tensor products}).

\subsection{Structure of the paper}
Sec.~\ref{sec: prelims} contains notation and relevant mathematical background. Sec.~\ref{sec: theory} contains an estimation schema based on linear programming that refines Fan's majorization relation. In Sec.~\ref{sec: applications to sums}, the schema is used to establish Eq.~\ref{eq: sep Fan majorization relations}. The paper is concluded in Sec.~\ref{sec: conclusion}. The appendix App.~\ref{sec: appendix} contains a deferred part of a proof that the linear programming bounds are tight in the case of simultaneously diagonalizable summands.

\section{Notation and background}
\label{sec: prelims}
We refer to a set of \(k\) elements as a \(k\)-set. For \(x \in \bR^{d}\), \(x^{\da}\) denotes the vector in \(\mathbb{R}^{d}\) whose coordinates are the coordinates of \(x\) albeit ordered (weakly) decreasingly. For \(k \in [d]\), the sum of the \(k\) largest coordinates of \(x\) is denoted by \(s_k(x) := \sum_{i=1}^k x_i^\da\).  For a subset \(T \subseteq [d]\), we write \(\one_T \in \bR^d\) to denote indicator function of \(T\). That is, the vector that satisfies for each \(i \in T, (\one_T)_i = 1\) and for each \(i \notin T, (\one_T)_i = 0\). A vector \(x \in \bR^d\) is said to majorize a vector 
\(y \in \bR^d\) if
\begin{align} 
\label{eq: majorization cond}
s_k(x) \geq  s_k(y)
\end{align}
for each \(k \in [d -1]\) and \(s_d(x) = s_d(y)\). This statement is denoted by \(y \preceq x\). Majorization is defined for self-adjoint operators spectrally (see~Sec.~7 of Ref.~\cite{Ando_1989}). If \(A, A'\) are self-adjoint operators on \(\bC^d\) and \(\la(A') \preceq \la(A)\), then \(A\) is said to majorize \(A'\), and we denote that by \(A' \preceq A\).  As shorthand, we write \(s_k(A) := \sum_{i=1}^k \la_i(A)\). As a relation on a real vector space, majorization is reflexive and transitive, that is, a preorder. See Ref.~\cite{Marshall2011} for a more comprehensive account of majorization theory.

The statement that a vector space \(W\) is a subspace of a vector space \(W'\) is denoted by \(W \leq W'\). We denote the dimension of a subspace \(W\) by \(\dim(W)\) and the orthogonal projection operator, or projector, onto \(W\) by \(P_W\).

Let \(A\) be a self-adjoint operator on \(\bC^d\). For \(k \in [d]\), Fan's maximum principle, proven in Ref.~\cite{Fan1949}, gives an extremal formula for the sum of the \(k\) largest eigenvalues of \(A\):
\begin{align}
\label{eq: Ky Fan maximum principle}
 s_k(A) =  \max_{\dim(W)=k}  \tr(P_W A).
\end{align}
Fan's majorization relation Eq.~\ref{eq: Ky Fan majorization relation with da} is an immediate consequence of this formula. Let  \(\xi_1 (A), \xi_2 (A),\ldots, \xi_d (A)\) denote orthonormal vectors in \(\bC^d\) such that for each \(i \in [d]\), \(A \xi_i (A) = \la_i (A) \xi_i (A)\). These vectors and the corresponding eigenvalues make up an orthonormal spectral decomposition of \(A\). For brevity, we use \(\Pi_i (A)\) to denote the rank-\(1\) projector \(\xi_i (A) \xi_i (A)^*\). For each \(\ell \in [d]\), we denote the linear span of the first \(\ell\) of these eigenvectors by
\begin{align}
    V_{\da \ell} (A) := \langle \{\xi_i (A)\}_{i=1}^\ell \rangle.
\end{align}
We refer to \(V_{\da \ell} (A)\) as a largest \(\ell\) eigenvalue space of 
\(A\). We denote the projector onto this space by \(P_{\da \ell} (A)\). It is useful to express the data of such a spectral decomposition by a complete flag \cite{Gillespie2019},
\begin{align}
\label{eq: subspace flag}
V_{\bullet} (A) := (0 \leq V_{\da 1} (A) \leq V_{\da 2} (A) \leq \cdots \leq V_{\da d} (A) = \bC^d).
\end{align}
If \(\la(A)\) contains repeated values, that is, \(A\) is degenerate, then \(V_{\bullet} (A)\) is not uniquely defined. Unless stated otherwise, we assume a choice is fixed throughout the analysis.

A linear program is an optimization problem whose objective function is linear and whose feasible set is a (convex) polyhedron. We write `s.t.' as shorthand for the phrase `subject to'.  See Ref.~\cite{Schrijver1998} for more details on the theory of linear programming.

Let \(\cT\) be a partially ordered set. A subset \(T \subseteq \mathcal{T}\) is called a downward closed set if for all \(x \in T\) and \(y \in \mathcal{T}\), \(y \leq x \implies y \in T.\) The dual notion, meaning under replacing the symbol \(\leq\) with the symbol \(\geq\), is called an upward closed set. The complement of a downward closed set is an upward closed set. For a set \(S \subseteq \mathcal{T}\), its downward closure \(\da S := \{ y \in \mathcal{T} \mid \exists \, x \in S, y \leq x\}\) is the smallest downward closed set that contains it. The upward closure \(\ua S\) is defined dually. See J.B. Nation's lecture notes on order theory \cite{Nation2017}.

\section{Estimation schema}
\label{sec: theory}
Let \(A_1, A_2\) be self-adjoint operators on \(\bC^d\) and fix \(k \in [d]\). Our aim is to find upper bounds on \(s_k(A_1 + A_2)\) that are generally tighter than \(s_k(A_1) + s_k(A_2)\), the upper bound from Fan's majorization relation Eq.~\ref{eq: Ky Fan majorization relation with da}.

For a \(k\)-dimensional subspace \(W \leq \bC^d\) and \(i \in [d]\), we denote the overlaps
\begin{align}
\label{def: overlap var} 
 x^{(1)}_{W,i} := \tr(P_W \Pi_i (A_1)), \;  x^{(2)}_{W,i} := \tr(P_W \Pi_i(A_2)),
\end{align}
and define the real vectors \(x_{W}^{(1)} := (x^{(1)}_{W,i})_{i = 1}^d, x_{W}^{(2)} := (x^{(2)}_{W,i})_{i = 1}^d\), and \(x_W := x^{(1)}_W \oplus x^{(2)}_W\). We rewrite \(\tr(P_W (A_1 + A_2))\) in terms of these vectors.
\begin{equation}
\label{eq: rewriting of tr form}
\begin{aligned}
\tr(P_W (A_1 + A_2)) &= \sum_{i=1}^d \la_i (A_1) x^{(1)}_{W,i} + \sum_{i=1}^d \la_i (A_2) x^{(2)}_{W,i} \\
&= (\la (A_1) \oplus \la(A_2))^T x_W. 
\end{aligned}
\end{equation}
It can be verified by inspection that \(x_W\) satisfies the linear constraints:
\begin{equation}
\begin{aligned}
\label{eq: basic conditions}
&x^{(1)}_{W,i}, x^{(2)}_{W,i} \geq 0, \forall \, i \in [d],\\ 
&x^{(1)}_{W,i}, x^{(2)}_{W,i} \leq 1, \forall \, i \in [d], \\ 
&\sum_{i=1}^d x^{(1)}_{W,i} = \sum_{i=1}^d x^{(2)}_{W,i} = k.
\end{aligned}
\end{equation}
We refer to these constraints as the \textit{basic constraints}. The set of points in \(\bR^{2d}\) that satisfy the basic constraints is a bounded polyhedron that we denote by \(\cP_0^{(k)}\). We can use this polyhedron to prove the inequality \(s_k(A_1+A_2) \leq s_k(A_1) + s_k(A_2)\) in a rather fancy way. Observe that
\begin{align}
\label{eq: estimate for linear form}
s_k (A_1 + A_2) &= \max_{\dim(W) = k} \tr(P_W (A_1 + A_2))\\
&= \max_{\dim(W) = k} (\la (A_1) \oplus \la(A_2))^T x_W \leq (\la (A_1) \oplus \la(A_2))^T \bar{x},
\end{align}
where \(\bar{x} \in \bR^{2d}\) is an optimal point for the linear program
\begin{equation}
\label{def: lin program}
\begin{aligned}
&\max \: (\la (A_1) \oplus \la(A_2))^T x\\
&\: \text{s.t.} \: x \in \cP_0^{(k)}.
\end{aligned}
\end{equation}
A routine calculation shows that \(\one_{[k]} \oplus \one_{[k]}\) is an optimal point.

To obtain sharper estimates, we need to consider more constraints. To that end, we introduce \textit{alignment terms}. 
\begin{definition}
\label{def: alignment term}
For each pair \((\ell_1, \ell_2) \in [d] \times [d]\), we define the alignment term
\begin{align}
\label{def:  alignment quantity}
\al_{(\ell_1, \ell_2)}^{(k)} (V_\bullet(A_1), V_\bullet(A_2)) :&= \max_{\dim(W)=k} \tr(P_W (P_{\da \ell_1} (A_1) + P_{\da \ell_2} (A_2))) \\
&= s_k (P_{\da \ell_1} (A_1) + P_{\da \ell_2} (A_2)). 
\end{align}
\end{definition}

These \(d^2\) positive numbers are meant to capture the overlaps between the largest eigenvalue spaces of \(A_1\) and the largest eigenvalue spaces of \(A_2\). The flag pair argument \((V_\bullet(A_1), V_\bullet(A_2))\) is included to emphasize the dependence of alignment terms on the assumed spectral decompositions of \(A_1\) and \(A_2\). When it is clear from context, we omit it and write \(\al_{(\ell_1, \ell_2)}^{(k)}\).

For any \(k\)-dimensional subspace \(W \leq \bC^d\), \(x_W\) satisfies, by definition, the linear constraint
\begin{align}
\label{ineq: alignment cons}
(\one_{[\ell_1]} \oplus \one_{[\ell_2]})^T x_W \leq \al_{(\ell_1, \ell_2)}^{(k)} (V_\bullet(A_1), V_\bullet(A_2)).
\end{align}
We call this an \textit{alignment constraint}. We denote the polyhedron of points in \(\cP_{0}^{(k)}\) that satisfy all alignment constraints by \(\cP_1^{(k)} (V_\bullet(A_1), V_\bullet(A_2))\) or simply \(\cP_1^{(k)}\).

\begin{thm}
\label{thm: linear programming bound}
Let \(u_k(A_1, A_2)\) denote the optimal value of the linear program 
\begin{equation}
\label{def: lin program for thm}
\begin{aligned}
&\max \: (\la (A_1) \oplus \la(A_2))^T x\\
&\: \text{s.t.} \: x \in \cP_1^{(k)} (V_\bullet(A_1), V_\bullet(A_2)).
\end{aligned}
\end{equation}
Then \(s_k (A_1 + A_2) \leq u_k(A_1, A_2)\).
\end{thm}
\begin{proof}
The bound follows from Fan's maximum principle Eq.~\ref{eq: Ky Fan maximum principle}, the fact that \(x_W \in \cP_1^{(k)} (V_\bullet(A_1) V_\bullet(A_2))\) for every \(k\)-dimensional subspace \(W\), and the equalities in Eq.~\ref{eq: rewriting of tr form}. \end{proof}

Several remarks are in order. 
\begin{enumerate}    
    \item Although \(\al_{(\ell_1, \ell_2)}^{(k)}\) is given by a maximization over \(k\)-dimensional subspaces, we find it is more amenable to analysis than \(s_k(A_1 + A_2)\). After all, the spectrum of a sum of two projectors is closely related to the relative position of their supports (see, for example, Ref.~\cite{Halmos1969}). This is a crucial ingredient in our proof of Prop.~\ref{prop: bound for alignment terms}. 
    
    \item If \(\cP \subseteq \bR^{2d}\) is a polyhedron that satisfies \(\cP_1^{(k)} (V_\bullet(A_1), V_\bullet(A_2)) \subseteq \cP \subseteq \cP_0^{(k)}\), then the optimal value of the linear program 
    \begin{equation}
    \label{def: lin program upper}
    \begin{aligned}
    &\max \: (\la (A_1) \oplus \la(A_2))^T x\\
    &\: \text{s.t.} \: x \in \cP.
    \end{aligned}
    \end{equation}
    is an upper bound on \(s_k (A_1 + A_2)\) that is at most \(s_k(A_1) + s_k(A_2)\). For example, \(\cP\) could be the polyhedron of points that satisfy the basic constraints and a subset of the alignment constraints. Or \(\cP\) could be equal to \(\cP_1^{(k)} (V_
    \bullet(\tilde{A_1}), V_\bullet (\tilde{A_2}))\) where \(\tilde{A_1}, \tilde{A_2}\) are self-adjoint operators on \(\bC^d\) such that 
    \begin{align}
    \label{eq: alignment comparison}
    \al_{(\ell_1, \ell_2)}^{(k)} (V_\bullet (A_1), V_\bullet(A_2)) \leq \al_{(\ell_1, \ell_2)}^{(k)} (V_\bullet (\tilde{A_1}), V_\bullet(\tilde{A_2}))
    \end{align}
    holds for each pair \((\ell_1, \ell_2) \in [d] \times [d]\). 
    
    \item The upper bound \(u_k (A_1, A_2)\) is in fact independent of the choice of spectral decompositions of \(A_1\) and \(A_2\). To see this, let us suppose that for some \(t \in [d-1]\), \(\la_t(A_1) = \la_{t+1}(A_1)\). Let \(U\) be a unitary operator on \(\bC^d\) that acts as the identity on  \(\langle \{ \xi_t (A_1), \xi_{t+1} (A_1) \} \rangle^{\perp}\) but is otherwise arbitrary. Let \(V'_\bullet (A_1)\) denote the flag
    \begin{align}
    \label{eq: flag action}
      (0 \leq U V_{\da 1} (A_1) \leq U V_{\da 2} (A_1) \leq \cdots \leq U V_{\da d} (A_1) = \bC^d).
    \end{align}
    Notice that for all \(\ell \neq t\), \(U V_{\da \ell} (A_1) = V_{\da \ell} (A_1)\) . This implies that for all \(\ell_1 \neq t\) and \(\ell_2\), it holds that
    \begin{gather}
    \al_{(\ell_1, \ell_2)}^{(k)} (V'_\bullet (A_1), V_\bullet(A_2)) = \al_{(\ell_1, \ell_2)}^{(k)} (V_\bullet (A_1), V_\bullet(A_2)).
    \end{gather}
    Now, suppose that \(x^{(1)} \oplus x^{(2)} \in \cP_1^{(k)} (V_\bullet(A_1), V_\bullet(A_2))\). Let \(x'^{(1)}\) denote the vector with coordinates 
    \begin{align}
    x'^{(1)}_{t} &:= x_t + x_{t+1} - \min(1,x_t + x_{t+1}) \\ 
    x'^{(1)}_{t+1} &:=  \min(1, x_t + x_{t+1})
    \end{align}
    and otherwise \(x'^{(1)}_i := x^{(1)}_i\). The objective value at \(x'^{(1)} \oplus x^{(2)}\) equals the objective value at \(x^{(1)} \oplus x^{(2)}\). It can be verified by inspection that \(x'^{(1)} \oplus x^{(2)}\) satisfies the basic constraints. For \(\ell_1 \neq t\) and every \(\ell_2\), 
    \begin{align}
        (\one_{[\ell_1]} \oplus \one_{[\ell_2]})^T (x'^{(1)} \oplus x^{(2)}) &= (\one_{[\ell_1]} \oplus \one_{[\ell_2]})^T (x^{(1)} \oplus x^{(2)}) \\
        &\leq \al_{(\ell_1, \ell_2)}^{(k)} (V_\bullet (A_1), V_\bullet(A_2)) \\
        &= \al_{(\ell_1, \ell_2)}^{(k)} (V'_\bullet (A_1), V_\bullet(A_2)).
    \end{align}
    As for the alignment constraints associated with pairs of the form \((t, \ell_2)\), we argue in two cases. First, if \(\min(1,x_t + x_{t+1}) = x_t + x_{t+1}\), then 
    \begin{align}
         (\one_{[t]} \oplus \one_{[\ell_2]})^T (x'^{(1)} \oplus x^{(2)}) = (\one_{[t]\setminus\{t\}} \oplus \one_{[\ell_2]})^T (x^{(1)} \oplus x^{(2)}).
    \end{align}
    If \(t = 1\), then 
    \begin{align}
    (\one_{[t]\setminus\{t\}} \oplus \one_{[\ell_2]})^T (x^{(1)} \oplus x^{(2)}) &= \one_{[\ell_2]}^T x^{(2)} \\
    &\leq \min(k, \ell_2) \\
    &\leq \al_{(1, \ell_2)}^{(k)} (V'_\bullet (A_1), V_\bullet(A_2)).
    \end{align}
    Otherwise, 
    \begin{align}
    (\one_{[t]\setminus\{t\}} \oplus \one_{[\ell_2]})^T (x^{(1)} \oplus x^{(2)}) &= (\one_{[t-1]} \oplus \one_{[\ell_2]})^T (x^{(1)} \oplus x^{(2)}) \\
    &\leq \al_{(t-1, \ell_2)}^{(k)} (V_\bullet (A_1), V_\bullet(A_2))\\
    &= \al_{(t-1, \ell_2)}^{(k)} (V'_\bullet (A_1), V_\bullet(A_2)) \\
    &\leq \al_{(t, \ell_2)}^{(k)} (V'_\bullet (A_1), V_\bullet(A_2)).
    \end{align}
Second, if \(\min(1,x_t + x_{t+1}) = 1\), then 
    \begin{align}
    (\one_{[t]} \oplus \one_{[\ell_2]})^T (x'^{(1)} \oplus x^{(2)}) &= (\one_{[t+1]} \oplus \one_{[\ell_2]})^T (x^{(1)} \oplus x^{(2)}) - 1 \\
    &\leq \al_{(t+1, \ell_2)}^{(k)} (V_\bullet (A_1), V_\bullet(A_2)) - 1 \\
    &= \al_{(t+1, \ell_2)}^{(k)} (V'_\bullet (A_1), V_\bullet(A_2)) - 1\\
    &\leq \al_{(t, \ell_2)}^{(k)} (V'_\bullet (A_1), V_\bullet(A_2)).
    \end{align}
    Hence, \(x'^{(1)} \oplus x^{(2)} \in \cP_1^{(k)} (V'_\bullet(A_1), V_\bullet(A_2))\). These considerations extend to cases where there are other degeneracies in either \(A_1\) or \(A_2\). 
\end{enumerate}

\subsection{Total unimodularity of the constraints}
\label{subsec: TU}
A useful observation about the alignment constraints is that vectors of the form \(\one_{[\ell_1]} \oplus \one_{[\ell_2]}\) can be made to have the \textit{consecutive ones property} by a single basis permutation. Let \(\sigma: [d] \rightarrow [d]\) denote the mirror permutation with the action \(\ell \mapsto d - \ell + 1\) for each \(\ell \in [d]\). Applying this permutation to the first \(d\) elements of the basis of \(\bR^{2d}\), corresponding to the first summand, yields 
\begin{align}
  \one_{[\ell_1]} \oplus \one_{[\ell_2]}\ \mapsto \one_{\sigma([\ell_1])} \oplus \one_{[\ell_2]} = (0, \ldots 0, \underbrace{1, \ldots 1}_{\ell_1 \text{times}} \mid \underbrace{1, \ldots 1,}_{\ell_2 \text{times}} 0, \ldots 0).  
\end{align}
The basic constraints are also given by vectors that have the consecutive ones property.

This is significant because a matrix with the consecutive ones property, also known as an interval matrix, is \textit{totally unimodular}. That means the determinant of every one of its square submatrices equals \(+1, 0,\) or \(-1\). This implies that if the alignment terms are integers, then the extremal points of \(\cP_1^{(k)}\) are integral, meaning their coordinates are integers (see~Ch.~19 of Ref.~\cite{Schrijver1998}). From the basic constraints, this means that their coordinates are equal to \(0\) or \(1\). Since a submatrix of a totally unimodular matrix is totally unimodular, the same is true for the extremal points of any polyhedron in \(\cP_0^{(k)}\) that is defined by a subset of the alignment constraints.

\subsection{The meaning of alignment terms}
\label{subsec: alignment terms}
We move on to discussing the meaning of alignment terms. To start, we observe that Fan's maximum principle Eq.~\ref{eq: Ky Fan maximum principle} implies the upper bound 
\begin{align}
\label{eq: upper bound for alpha}
\al_{(\ell_1, \ell_2)}^{(k)} \leq \min(\ell_1, k) + \min(\ell_2, k),
\end{align}
which is satisfied with equality if \(P_{\da \ell_1} (A_1) \) and \(P_{\da \ell_2} (A_2)\) are perfectly aligned. We observe that Lidskii's majorization relation (see Sec.~III.4 of Ref.~\cite{Bhatia1997}) implies the lower bound
\begin{align}
\label{eq: lower bound for alpha}
\al_{(\ell_1, \ell_2)}^{(k)} \geq \min ( k, \max(0, \ell_1 + \ell_2 - d)) + \min(k, \ell_1 + \ell_2). 
\end{align}
When \(\ell_1 + \ell_2 \leq k\), the upper bound and lower bound coincide, implying the identity \(\al_{(\ell_1, \ell_2)}^{(k)} = \ell_1 + \ell_2\). This says that the cases where some information about the operators \(A_1\) and \(A_2\) is retained are the ones where \(\ell_1 + \ell_2 > k\).

If the upper bound Eq.~\ref{eq: upper bound for alpha} is satisfied with equality, then the alignment constraint corresponding to the pair \((\ell_1, \ell_2)\) is redundant within \(\cP_0^{(k)}\). On the other hand, the slackness statement  \(\lceil \al_{(\ell_1, \ell_2)}^{(k)} \rceil < \min(\ell_1, k) + \min(\ell_2, k)\) implies an upper bound on \(s_k (A_1 + A_2)\) with a combinatorial meaning. In words, \(\lceil \al_{(\ell_1, \ell_2)}^{(k)} \rceil - k\) is the maximum number of eigenvalue pairs whose indices are in \([\ell_1]\) and \([\ell_2]\), respectively, that can appear (with coefficient \(1\)) together in the objective \((\la (A_1) \oplus \la(A_2))^T x\). The remaining eigenvalue pairs have to be staggered. Notice that, per the discussion in the previous paragraph, the slackness statement implies \(\ell_1 + \ell_2 > k\).

\begin{prop}
\label{prop: inequality for s_k based on ell_1, ell_2}
Let \(m := \lceil \al_{(\ell_1, \ell_2)}^{(k)} \rceil - k\) and suppose that it satisfies \(m < \min(\ell_1, k) + \min(\ell_2, k) - k\). Denote the tuple of staggered eigenvalue sums
\begin{align}
\label{def: staggered vec}
   \mu := (\la_i (A_1) + \la_{i - m + \ell_2} (A_2))_{i=m+1}^k \oplus (\la_{i-m+\ell_1} (A_1) + \la_i(A_2))_{i=m+1}^k.
\end{align}
Then \(s_k (A_1 + A_2) \leq s_{m} (A_1) + s_{m} (A_2) + s_{k - m} (\mu)\).
\end{prop}

\begin{proof}
By Thm.~\ref{thm: linear programming bound}, it suffices to show that the optimal value of the linear program
\begin{equation}
\label{def: lin prog for stagg}
\begin{aligned}
&\max \: (\la (A_1) \oplus \la(A_2))^T x\\
&\: \text{s.t.} \: x \in \cP_0^{(k)}, \:(\one_{[\ell_1]} \oplus \one_{[\ell_2]})^T x \leq m+k
\end{aligned}
\end{equation}
is \(s_{m} (A_1) + s_{m} (A_2) + s_{k-m} (\mu)\). By total unimodularity, the extremal points of the feasible set of this program are integral (see Sec.~\ref{subsec: TU}). Hence, it suffices to maximize \((\la (A_1) \oplus \la(A_2))^T (\one_{S_1} \oplus \one_{S_2})\) over subset pairs \((S_1, S_2)\) that satisfy \(S_1, S_2 \subseteq [d], |S_1| = |S_2|=k\) and \(| S_1 \cap [\ell_1]| + | S_2 \cap [\ell_2]| \leq m + k.\) 

Given such a \(k\)-set pair \((S_1, S_2)\), we may assume without loss of generality that \(S_1 \cap [\ell_1] = [r_1]\) and \(S_2 \cap [\ell_2] = [r_2]\) where \(r_1 := |S_1 \cap [\ell_1]|\) and \(r_2 := |S_2 \cap [\ell_2]|\). Otherwise, the intersections \(S_1 \cap [\ell_1]\) and \(S_2 \cap [\ell_2]\) may be adjusted to improve the objective value. By the same token, we may assume that \(S_1 \cap [\ell_1]^{c} = \{\ell_1 + 1, \ldots,\ell_1 + k - r_1\}\) and \(S_2 \cap [\ell_2]^c = \{\ell_2 + 1, \ldots, \ell_2 + k - r_2\}\).  Lastly, we argue that the alignment constraint may be assumed to be satisfied with equality. If \(| S_1 \cap [\ell_1]| + | S_2 \cap [\ell_2]| \neq m + k\), then \(S_1\) may be adjusted to improve the objective value by replacing elements in \(S_1 \cap [\ell_1]^{c}\) with elements in \([\ell_1]\) until the constraint is satisfied with equality or \(S_1 \cap [\ell_1]^{c}\) becomes empty. If the latter happens before the former, \(S_2\) may be adjusted to improve the objective value by replacing elements in \(S_2 \cap [\ell_2]^c\) with elements in \([\ell_2]\) until the constraint is satisfied with equality. Hence, without loss of generality, we assume that \(r_1 + r_2 = m + k\). From the fact that \(\max(r_1, r_2) \leq k\), we conclude that \(m \leq \min(r_1, r_2)\).

With these assumptions on \((S_1, S_2)\), the objective function is
\begin{equation}
\begin{aligned}
    (\la (A_1) \oplus \la(A_2))^T (\one_{S_1} \oplus \one_{S_2}) &= s_{m} (A_1) + s_{m} (A_2) \\
    &+\sum_{i=m+1}^{r_1} \la_i(A_1) +  \sum_{i=1}^{r_1-m} \la_{i+\ell_2}(A_2) \\
    &+  \sum_{i=1}^{k-r_1} \la_{i+\ell_1}(A_1) + \sum_{i=m+1}^{m+k-r_1} \la_i (A_2).
\end{aligned}
\end{equation}
We rewrite the sums in the second and third lines.
\begin{equation}
\begin{aligned}
   &\sum_{i=m+1}^{r_1} \la_i(A_1) +  \sum_{i=1}^{r_1-m} \la_{i+\ell_2}(A_2) = \sum_{i=m+1}^{r_1} \la_i(A_1) +  \la_{i-m+\ell_2}(A_2), \\
   &\sum_{i=1}^{k-r_1} \la_{i+\ell_1}(A_1) + \sum_{i= m+1}^{m+k-r_1} \la_i (A_2)  =  \sum_{i= m+1}^{m+k-r_1}  \la_{i-m+\ell_1}(A_1) + \la_i (A_2) .
\end{aligned}
\end{equation}
Then, by the definition of \(\mu\), it holds that 
\begin{equation}
\begin{aligned}
\sum_{i=m+1}^{r_1} (\la_i(A_1) +  \la_{i-m+\ell_2}(A_2)) &+ \sum_{i= m+1}^{m+k-r_1}  (\la_{i-m+\ell_1}(A_1) + \la_i (A_2))\\
&\leq s_{k-m} (\mu).
\end{aligned}
\end{equation}
Equality is realized by maximizing over \(r_1\). \end{proof}

\subsection{Simultaneously diagonalizable summands}
\label{subsec: tightness}
In this subsection,  we outline a strategy for proving that the linear programming upper bound given by Thm.~\ref{thm: linear programming bound} is tight in the case of simultaneously diagonalizable summands. We complete the proof in App.~\ref{sec: appendix}.

Let \(D_1\) and \(D_2\) be two simultaneously diagonalizable self-adjoint operators on \(\bC^d\). Without loss of generality, we assume that \(D_1\) and \(D_2\) are diagonal in the standard basis \(\{q_i\}_{i=1}^d\). Furthermore, we take \(V_\bullet (D_1)\) and \(V_\bullet(D_2)\) to be complete flags corresponding to diagonal spectral decompositions of \(D_1\) and \(D_2\), respectively. For each \(i \in [d]\), we denote \(\tilde{\la}_i (D_1) := q_i^* D_1 q_i\) and \(\tilde{\la}_i(D_2) := q_i^* D_2 q_i\). To capture the order of the eigenvalues, we define the sets
\begin{gather}
\Og^1_\ell := \{ i \in [d] \mid P_{\da \ell} (D_1) q_i = q_i\}, \Og^2_\ell := \{ i \in [d] \mid P_{\da \ell} (D_2) q_i = q_i\}.
\end{gather}
for each \(\ell \in [d]\). These are the sets of coordinates of the largest \(\ell\) values in \(\tilde{\la} (D_1) := (\tilde{\la}_i(D_1))_{i=1}^d\) and \(\tilde{\la} (D_2) := (\tilde{\la}_i (D_2))_{i=1}^d\), respectively. We use the two subset chains \(\{ \Og^{1}_{\ell_1}\}_{\ell_1 =1}^d\) and \(\{ \Og^{2}_{\ell_2}\}_{\ell_2 =1}^d\) to define two total orders, \(\geq^1\) and \(\geq^2\), on \([d]\). For \(x, x' \in [d]\), we write \(x \geq^1 x'\) if  
\begin{align}
\label{def: induced total order}
\min \{ \ell \mid x \in \Og^1_\ell\} \geq \min \{ \ell \mid x' \in \Og^1_\ell\}.
\end{align}
Similarly, we write \(x \geq^2 x'\) if
\begin{align}
\min \{ \ell \mid x \in \Og^2_\ell\} \geq \min \{ \ell \mid x' \in \Og^2_\ell\}.
\end{align}
Observe that \(x \geq^1 x'\) implies \(\tilde{\la}_x (D_1) \leq \tilde{\la}_{x'} (D_1)\) and \(x \geq^2 x'\) implies \(\tilde{\la}_x (D_2) \leq \tilde{\la}_{x'} (D_2)\).

Now, we compute the alignment terms of \(D_1\) and \(D_2\).
\begin{lem}
\label{lem: explicit form for alignment terms}
For each \( (\ell_1, \ell_2) \in [d] \times [d]\), it holds that
\begin{gather}
\al_{(\ell_1, \ell_2)}^{(k)} (V_{
\bullet} (D_1), V_{\bullet} (D_2))
= \min(k, | \Og^{1}_{\ell_1} \cap \Og^{2}_{\ell_2} |) + \min(k, |\Og^{1}_{\ell_1} \cup \Og^{2}_{\ell_2}|).
\end{gather}
\end{lem}
\begin{proof}
Recall that 
\begin{align}
 \al_{(\ell_1, \ell_2)}^{(k)} (V_{\bullet} (D_1), V_{\bullet} (D_2)) &= \max_{\dim(W) =k} \tr (P_W ( P_{\da \ell_1} (D_1) + P_{\da \ell_2} (D_2)))\\
 &= s_k( P_{\da \ell_1} (D_1) + P_{\da \ell_2} (D_2)) \\
 &= s_k( \sum_{i \in \Og^1_{\ell_1}} q_i q_i^* + \sum_{i \in \Og^2_{\ell_2}} q_i q_i^*). 
\end{align}
If \(k \leq |\Og^{1}_{\ell_1} \cap \Og^{2}_{\ell_2}|\), then it equals \(2k\). If \(|\Og^{1}_{\ell_1} \cap \Og^{2}_{\ell_2}| < k < |\Og^{1}_{\ell_1} \cup \Og^{2}_{\ell_2}|\), then it equals \( 2 |\Og^{1}_{\ell_1} \cap \Og^{2}_{\ell_2}| + k - |\Og^{1}_{\ell_1} \cap \Og^{2}_{\ell_2}| = |\Og^{1}_{\ell_1} \cap \Og^{2}_{\ell_2}| + k\). Lastly, if \(k \geq |\Og^{1}_{\ell_1} \cup \Og^{2}_{\ell_2}|\), then it equals \(|\Og^{1}_{\ell_1}| + |\Og^2_{\ell_2}|\) which equals \( |\Og^{1}_{\ell_1} \cap \Og^{2}_{\ell_2}| + |\Og^{1}_{\ell_1} \cup \Og^{2}_{\ell_2}|\). \end{proof}

Since the alignment terms are integers, the coordinates of the extremal points of \(\cP_1^{(k)} (V_{\bullet} (D_1), V_{\bullet} (D_2))\) are in \(\{0,1\}\). Therefore, the linear programming upper bound \(u_k (D_1, D_2)\) equals the optimal value of the following combinatorial optimization problem:
\begin{align}
\label{eq: obj comb}
\max \; (\tilde{\la} (D_1) \oplus \tilde{\la}(D_2))^T (\one_{S_1} \oplus \one_{S_2})
\end{align}
over subsets pairs \((S_1, S_2)\) satisfying \(S_1, S_2 \subseteq [d], |S_1| = |S_2| = k\), and for all \((\ell_1, \ell_2) \in [d] \times [d]\),
\begin{align}
\label{eq: comb align cons}
|S_1 \cap \Og^1_{\ell_1}| + |S_2 \cap \Og^2_{\ell_2}| \leq \min(k, | \Og^{1}_{\ell_1} \cap \Og^{2}_{\ell_2} |) + \min(k, |\Og^{1}_{\ell_1} \cup \Og^{2}_{\ell_2}|).
\end{align}
We call a subset pair \((S_1, S_2)\) that satisfies all of these constraints a \textit{feasible pair}. Notice that the objective here differs from the one in Thm.~\ref{thm: linear programming bound} by a relabeling of the coordinates of the linear form.

Next, we observe that \(s_k(D_1 + D_2)\) equals 
\begin{align}
\label{eq: observation about symmetric set pair}
\max_{S \subseteq [d], |S|=k} \sum_{i \in S} q_i^* (D_1 +  D_2) q_i = \max_{S \subseteq [d], |S|=k} \sum_{i \in S} \tilde{\la}_i (D_1) + \tilde{\la}_i (D_2)
\end{align}
which can be rewritten as follows
\begin{align}
\max_{(S, S), |S|=k} \; (\tilde{\la} (D_1) \oplus \tilde{\la} (D_2))^T (\one_{S} \oplus \one_{S}).
\end{align}                                 
We call set pairs of the form \((S, S)\) \textit{symmetric}. Therefore, to prove that \(u_k (D_1, D_2) = s_k (D_1 + D_2)\), it suffices to show that for an arbitrary feasible pair \((S_1, S_2)\), there exists a feasible symmetric pair \((S, S)\) such that the objective Eq.~\ref{eq: obj comb} at \((S,S)\) is greater than or equal to the objective at \((S_1, S_2)\). 

\begin{thm}
\label{thm: tightness for diagonal operators}
It holds that 
\begin{gather}
u_k (D_1, D_2) = s_k (D_1 + D_2).    
\end{gather}
\end{thm}
\begin{proof}
Deferred to appendix App.~\ref{sec: appendix}. 
\end{proof}

\section{Applications to sums of tensor products}
\label{sec: applications to sums}
We use theorems Thm.~\ref{thm: linear programming bound} and Thm.~\ref{thm: tightness for diagonal operators} to prove our main theorem Thm.~\ref{thm: sep Fan majorization}.

Fix two positive integers \(d_B\) and  \(d_C\). Let \(B_1, B_2\) be positive semi-definite operators on \(\bC^{d_B}\) and \(C_1, C_2\) be positive semi-definite operators on \(\bC^{d_C}\). The goal is to prove the majorization relation
\begin{align}
\label{rel: sep fan again}
B_1 \otimes C_1 + B_2 \otimes C_2 \preceq B_1^{\da} \otimes C_1^{\da} + B_2^{\da} \otimes C_2^{\da}. 
\end{align}
We start with an example showing that the assumption of positive semi-definiteness cannot be relaxed in general. 
\begin{example}
\label{ex: counter example with negative eigenvalues}
Consider the following operators on $\bC^2$: \(\tilde{B}_{1} := q_1 q_1^* + q_2 q_2^* =: - \tilde{B}_{2}, \tilde{C}_1 := q_{1} q_{1}^*,\) and \(\tilde{C}_2 := q_{2} q_{2}^*\). The spectrum of the sum \(\tilde{B}_1 \otimes \tilde{C}_1 + \tilde{B}_2 \otimes \tilde{C}_2\) is \((1,1,-1,-1)\). On the other hand, the aligned sum \(\tilde{B}_1^\da \otimes \tilde{C}_1^\da + \tilde{B}_2^{\da} \otimes \tilde{C}_2^{\da}\) is null, and so its spectrum is \((0,0,0,0) \not\succeq (1,1,-1,-1)\).
\end{example}

\subsection{Ordering the spectrum}
\label{subsec: order of spectrum of product} 
Let \(\tilde{B}\) be a positive semi-definite operator on \(\bC^{d_B}\) and \(\tilde{C}\) be a positive semi-definite operator on \(\bC^{d_C}\). Their tensor product \(\tilde{B} \otimes \tilde{C}\) has the spectral decomposition
\begin{gather}
\sum_{(i,j) \in [d_B] \times [d_C]} \la_i (\tilde{B}) \la_j (\tilde{C}) \Pi_i (\tilde{B}) \otimes \Pi_j (\tilde{C}).
\end{gather}
The order of the eigenvalues of this operator is not determined entirely by the orders of the eigenvalues of \(\tilde{B}\) and \(\tilde{C}\). In general, for \((i, j), (i', j') \in [d_B] \times [d_C]\), whether the inequality 
\begin{gather}
\la_i (\tilde{B}) \la_j (\tilde{C}) \leq \la_{i'} (\tilde{B}) \la_{j'} (\tilde{C})
\end{gather}
holds or not depends on the eigenvalues in question, not only their order. Here is what we can conclude however. Since the eigenvalues are non-negative, if \(i \geq i'\) and \(j \geq j'\), then the inequality holds. That is to say, the orders of the eigenvalues of \(\tilde{B}\) and \(\tilde{C}\) induce a partial order on the eigenvalues of \(\tilde{B} \otimes \tilde{C}\), namely the product order. We denote this relation by \(\geq^\times\).

To account for the remaining order relations implied by the spectra of \(\tilde{B}\) and \(\tilde{C}\), we associate with \(\tilde{B} \otimes \tilde{C}\) a chain of downward closed sets in \(([d_B] \times [d_C], \geq^\times)\) 
\begin{align}
\label{chain}
\Ups_{1} = \{(1,1)\} \subset \Ups_{2} \subset \cdots \subset \Ups_{d_B d_C} = [d_B] \times [d_C]    
\end{align} 
such that for each \(\ell \in [d_B d_C]\), we have 
\begin{align}
\label{eq: downward closed projector}
V_{\da \ell} (\tilde{B} \otimes \tilde{C}) = \langle \{ \xi_i (\tilde{B}) \otimes \xi_j (\tilde{C}) \mid (i, j) \in \Ups_{\ell}\} \rangle.
\end{align}
Observe that this implies \(|\Ups_{\ell}| = \ell\).

That such a chain exists can be proven by induction. The chain in Eq.~\ref{chain} induces a total order on \([d_B] \times [d_C]\). For \((i,j), (i', j') \in [d_B] \times [d_C]\), \((i',j')\) precedes \((i, j)\) in this order if 
\begin{align}
\min \{ \ell \mid (i, j) \in \Ups_\ell\} \geq \min \{ \ell \mid (i', j') \in \Ups_\ell\}.
\end{align}
Observe that if this holds, then \(\la_i (\tilde{B}) \la_j (\tilde{C}) \leq \la_{i'} (\tilde{B}) \la_{j'} (\tilde{C})\). The figure Fig.~\ref{fig: Hasse diagram for 2by2} depicts a Hasse diagram for \(([2] \times [2], \geq^\times)\) along with a total order induced by a chain of downward closed sets.

Let \(\{ \Ups^{1}_{\ell} \}_{\ell=1}^{d_B d_C}\) and \(\{ \Ups^{2}_{\ell} \}_{\ell=1}^{d_B d_C}\) denote the two chains of downward closed sets associated with \(B_1 \otimes C_1\) and \(B_2 \otimes C_2\), respectively.

\begin{figure}[t]
\centering
\begin{tikzpicture}[scale=.5]
  \node (top) at (0,2.5) {\((2,2)\)};
  \node (a) at  (3,0) {\((1,2)\)};
  \node (d) at (-3,0) {\((2,1)\)};
  \node (bottom) at (0,-2.5) {\((1,1)\)};
  \draw (bottom)--(a);
  \draw (bottom) -- (d);
  \draw (top) -- (a);  
  \draw (top) -- (d);
  \draw [->, violet, opacity=0.8, very thick] plot [smooth, tension=1] coordinates { (0,-2) (2,0) (-2,0) (0,2)};
\end{tikzpicture}
\caption{{A Hasse diagram for \(([2] \times [2], \geq^\times)\). The curved arrow indicates a total order on \([2] \times [2]\) induced by the chain \(\{(1,1)\} \subset \{(1,1), (1,2)\} \subset \{(1,1),(1,2),(2,1)\} \subset \{(1,1), (1,2),(2,1),(2,2)\}\). According to this total order, \((1,2)\) precedes \((2,1)\), though the two are incomparable in the product order. This is the order of the eigenvalues of \(\tilde{B} \otimes \tilde{C}\) when, for example, \(\tilde{B} = 3 q_1 q_1^* + q_2 q_2^*\) and \(\tilde{C} = 3 q_1 q_1^* + 2 q_2 q_2^*\).}}
\label{fig: Hasse diagram for 2by2}
\end{figure}
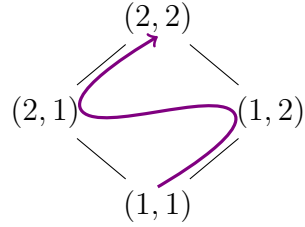

\subsection{Bounding the alignment terms} 
\label{subsec: bounding alignment}

To bound the alignment term \(\al_{(\ell_1, \ell_2)}^{(k)} (V_\bullet (B_1 \otimes C_1), V_
\bullet(B_2 \otimes C_2))\), we examine spaces spanned by product vectors indexed by elements of downward closed sets. Upon our examination, we find a generalization of a well known fact about subspace intersections. Specifically, that for any two orthonormal bases \(\{f_i\}_{i=1}^d\) and \(\{g_i\}_{i=1}^d\) for \(\bC^d\), and for any two downward closed sets \(Y, Y' \subseteq [d]\) (that is, both are intervals that start at \(1\)), it holds that 
\begin{align}
\label{eq: dim inequality for line}
\dim ( \langle \{ f_i\}_{i \in Y} \rangle \cap \langle \{g_i\}_{i \in Y'} \rangle^{\perp}) \geq | Y \setminus Y'|.
\end{align}

Let each of \(\{f_i^{(B)}\}_{i=1}^{d_B}\) and \(\{ g_i^{(B)}\}_{i=1}^{d_B}\) be an orthonormal basis for \(\bC^{d_B}\), and each of \(\{f_j^{(C)}\}_{j=1}^{d_C}\) and \(\{ g_j^{(C)}\}_{j=1}^{d_C}\) be an orthonormal basis for \(\bC^{d_C}\). For the sake of readability, we denote for \(i \in [d_B]\),
\begin{align}
F^{(B)}_{\da i} := \langle \{f^{(B)}_1, \ldots, f^{(B)}_{i}\} \rangle \, , \, G^{(B)}_{\ua i} := \langle \{g_i^{(B)}, \ldots, g_{d_B}^{(B)}\} \rangle,   
\end{align}
for \(j \in [d_C]\), 
\begin{align}
F^{(C)}_{\da j} := \langle \{f^{(C)}_1, \ldots, f^{(C)}_{j}\} \rangle \, , \, G^{(C)}_{\ua j} := \langle \{ g^{(C)}_{j}, \ldots, g^{(C)}_{d_C} \} \rangle,
\end{align}
and, finally, for a subset \(S \subseteq [d_B] \times [d_C]\),
\begin{align}
F_S := \langle \{ f_i^{(B)} \otimes f_j^{(C)} \mid (i, j) \in S \}\rangle, \\
G_S := \langle \{ g_i^{(B)} \otimes g_j^{(C)} \mid (i, j) \in S \}\rangle.
\end{align}

\begin{prop}
\label{prop: orthogonal tensor products} 
Let \(\Ups\) and \(\Ups'\) be two downward closed sets in \(([d_B] \times [d_C], \geq^\times)\). Then
\begin{align}
\label{eq: dim inequality for plane}
\dim(F_{\Ups} \cap {G_{\Ups'}}^\perp) \geq |\Ups \setminus \Ups'|.    
\end{align}
\end{prop}

\begin{proof} 
Observe that \({G_{\Ups'}}^\perp = G_{{\Ups'}^c}\). Denote the difference \(\Delta := \Ups \setminus \Ups'\). Since \(\da \Delta \subseteq \Ups\) and \(\ua \Delta \subseteq {\Ups'}^c\), 
it suffices to prove \(\dim ( F_{\da \Delta} \cap G_{\ua \Delta}) \geq |\Delta|.\) We do this by showing that there exists a subspace whose dimension is not so large, yet the dimensions of its intersections with \(F_{\da \Delta}\) and \(G_{\ua \Delta}\) are relatively large. 

As the difference of two downward closed sets, \(\Delta\) is interval-like. That is, if \(x,y \in \Delta\) and \(x \leq^\times z \leq^\times y\) for some \(z \in [d_B] \times [d_C]\), then \(z \in \Delta\). This implies \({\da \Delta} \cap {\ua \Delta} = \Delta\). Let \(\Delta_1 := \{i_1, \ldots, i_{T}\} \subseteq [d_B], \Delta_2 := \{ j_1, \ldots, j_{L}\} \subseteq [d_C]\) denote the sets of first coordinates and second coordinates, respectively, of the elements of \(\Delta\).  Let \((\tilde{i},\tilde{j})\) be an element of \(\Delta_1 \times \Delta_2\). There exists \(j' \in [d_C]\) such that \((\tilde{i},j') \in \Delta\). If \(j' \geq \tilde{j}\), then \((\tilde{i},\tilde{j}) \in  {\da \Delta}\). Otherwise, \( (\tilde{i}, \tilde{j}) \in {\ua \Delta}\). Therefore, the inclusion \(\Delta_1 \times \Delta_2 \subseteq  {\da \Delta} \cup {\ua \Delta}\) holds. See Fig.~\ref{fig: Hasse diagram for 3by3 with delta} for an illustrative example. 

\begin{figure}
\centering
\begin{tikzpicture}[scale=.5]
  \node (22) at (0,3) {\((2,2)\)};
  \node (12) at  (2.7,0) {\((1,2)\)};
  \node (21) at (-2.7,0) {\((2,1)\)};
  \node (11) at (0,-2.5) {\((1,1)\)};
  \node (13) at (5,3) {\(\mathbf{(1,3)}\)};
  \node (31) at (-5,3) {\(\mathbf{(3,1)}\)};
  \node (23) at (2.7,6) {\((2,3)\)};
  \node (32) at (-2.7,6) {\((3,2)\)};
  \node (33) at (0, 8.5) {\((3,3)\)};
  \draw[brown, very thick, <-] (11) -- (12);
  \draw[brown, very thick, <-] (11) -- (21);
  \draw (12) -- (22);  
  \draw (21) -- (22);
  \draw[brown, very thick, <-] (12) -- (13);
  \draw[brown, very thick, <-] (21) -- (31);
  \draw (22) -- (23);
  \draw[blue, very thick, ->] (13) -- (23);
  \draw (22) -- (32);
  \draw[blue, very thick, ->] (31) -- (32);
  \draw[blue, very thick, ->] (23) -- (33);
  \draw[blue, very thick, ->] (32) -- (33);
  \node[circle] (c) at (0,0){};  
\end{tikzpicture}
\caption{\small{A Hasse diagram for \(([3] \times [3], \geq^\times)\). In this example, the set \(\Delta\) (in bold) is \(\{(1,3), (3,1)\}\). The downward closure \({\da\Delta} = \{ (1,3), (1,2), (1,1), (2,1), (3,1)\}\) is indicated by the downward facing arrows, while the upward closure \({\ua\Delta} = \{ (1,3),(2,3), (3,3), (3,2), (3,1)\}\) is indicated by the upward facing arrows. In this example, the product \(\Delta_1 \times \Delta_2\) is  \(\{ (1,3), (1,1), (3,1), (3,3)\}\).}} 
\label{fig: Hasse diagram for 3by3 with delta}
\end{figure}
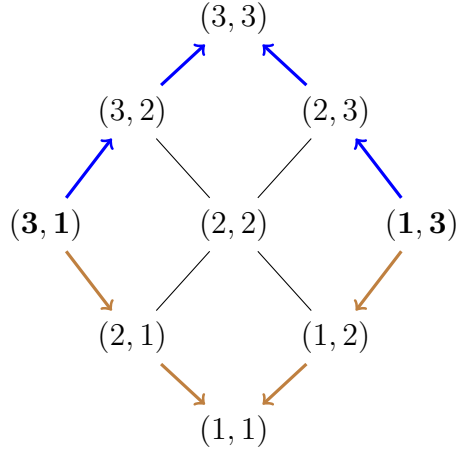

By a dimension counting argument due to Amir-Moéz (see Lem~6.1 in Ref.~\cite{Ando_1989} and Sec.~III.3 in Ref.~\cite{Bhatia1997}), there exist an orthonormal set \(\{\tilde{f}^{(B)}_{i_t}\}_{t=1}^T\)  satisfying 
\(\tilde{f}^{(B)}_{i_t} \in F^{(B)}_{\da i_t}\) for each \(t \in [T]\) and an orthonormal set \(\{ \tilde{g}^{(B)}_{i_t} \}_{t=1}^T\) satisfying 
\(\tilde{g}^{(B)}_{i_t} \in G^{(B)}_{\ua i_t}\) for each \(t \in [T]\) such that 
\begin{align}
\label{eq: W_1 }
\langle \{\tilde{f}^{(B)}_{i_t}\}_{t=1}^T \rangle = \langle \{ \tilde{g}^{(B)}_{i_t}\}_{t=1}^T \rangle=: W^{(B)}  
\end{align}
By the same token, there exist an orthonormal set \(\{\tilde{f}^{(C)}_{j_\ell}\}_{\ell=1}^L\)  satisfying \(\tilde{f}^{(C)}_{j_\ell} \in F^{(C)}_{\da j_\ell}\) for each \(\ell \in [L]\) and an orthonormal set \(\{ \tilde{g}^{(C)}_{j_\ell} \}_{\ell=1}^L\) satisfying \(\tilde{g}^{(C)}_{j_\ell} \in G^{(C)}_{\ua j_\ell}\) for each \(\ell \in [L]\) such that 
\begin{align}
\label{eq: W_2}
\langle \{\tilde{f}^{(C)}_{j_\ell}\}_{\ell=1}^L \rangle = \langle \{ \tilde{g}^{(C)}_{j_\ell}\}_{\ell=1}^L \rangle=: W^{(C)}.
\end{align}
The subspace we seek is the tensor product \(W := W^{(B)} \otimes W^{(C)}\).

For \((i, j) \in {\da \Delta}\), the inclusion \(\da\{ (i, j)\} \subseteq {\da \Delta}\) implies \(F_{\da \{(i,j)\}} \leq F_{\da \Delta}\). Since \(F_{\da i}^{(B)} \otimes F_{\da j}^{(C)} = F_{\da \{(i,j)\}}\), the membership \((i_t, j_\ell) \in (\Delta_1 \times \Delta_2) \cap {\da \Delta} \) implies \(\tilde{f}^{(B)}_{i_t} \otimes \tilde{f}^{(C)}_{j_\ell} \in F_{\da \Delta}\). This implies 
\begin{gather}
 \dim(F_{\da \Delta} \cap W) \geq |(\Delta_1 \times \Delta_2) \cap {\da \Delta}|.   
\end{gather}
Similar considerations for the other basis yield  
\begin{gather}
\dim( G_{\ua \Delta} \cap W) \geq | (\Delta_1 \times \Delta_2) \cap {\ua \Delta}|.     
\end{gather}
Putting it altogether, we obtain the following estimate.
\begin{align}
   \dim (F_{\da \Delta} \cap G_{\ua \Delta}) &\geq \dim ((F_{\da \Delta} \cap W) \cap (G_{\ua \Delta} \cap W)) \notag\\
   &\geq \dim(F_{\da \Delta} \cap W) + \dim(G_{\ua \Delta} \cap W) - \dim(W) \notag \\
   &\geq |(\Delta_1 \times \Delta_2) \cap {\da \Delta}| + | (\Delta_1 \times \Delta_2) \cap {\ua \Delta}| - |\Delta_1 \times \Delta_2| \notag \\
   &= |\Delta|, 
\end{align} 
where for the equality we used inclusion-exclusion and the facts \(\Delta_1 \times \Delta_2 \subseteq  {\da \Delta} \cup {\ua \Delta}\) and \({\da \Delta} \cap {\ua \Delta} = \Delta \subseteq \Delta_1 \times \Delta_2\). \end{proof}

With this in hand, we can obtain tight upper bounds for the alignment terms of tensor products of positive semi-definite operators. 

\begin{prop}
\label{prop: bound for alignment terms}
For each pair \((\ell_1, \ell_2) \in [d_B d_C] \times [d_B d_C]\),
\begin{align*}
\al_{(\ell_1, \ell_2)}^{(k)} (V_{\bullet}(B_1 \otimes C_1), V_{\bullet}(B_2 \otimes C_2)) \leq \al_{(\ell_1, \ell_2)}^{(k)} (V_{
\bullet} (B_1^{\da} \otimes C_1^{\da}), V_{\bullet} (B_2^{\da} \otimes C_2^{\da})).
\end{align*}
\end{prop}
\begin{proof}
By the previous proposition, \(V_{\da \ell_1} (B_1 \otimes C_1)\) contains a subspace of dimension \(|\Ups^{1}_{\ell_1} \setminus \Ups^{2}_{\ell_2}|\) that is orthogonal to \(V_{\da \ell_2} (B_2 \otimes C_2)\). By the same token, \(V_{\da \ell_2} (B_2 \otimes C_2)\) contains a subspace of dimension \(|\Ups^{2}_{\ell_2} \setminus \Ups^{1}_{\ell_1}|\) that is orthogonal to \(V_{\da \ell_1} (B_1 \otimes C_1)\). Hence, the multiplicity of \(1\) in the spectrum of \(P_{\da \ell_1} (B_1 \otimes  C_1) + P_{\da \ell_2} (B_2 \otimes C_2)\) is at least \(|\Ups^{1}_{\ell_1} \setminus \Ups^{2}_{\ell_2}| +  |\Ups^{2}_{\ell_2} \setminus \Ups^{1}_{\ell_1}|\). The sum of remainder of the eigenvalues is \(2 | \Ups^{1}_{\ell_1} \cap \Ups^{2}_{\ell_2}|\), and so  
\begin{align}
\label{eq: rel for sum of proj}
\la(P_{\da \ell_1} (B_1 \otimes  C_1) + P_{\da \ell_2} (B_2 \otimes C_2)) \preceq (2, \ldots,2, 1, \ldots, 1, 0, \ldots, 0),
\end{align}
where on the right hand side the multiplicity of \(2\) is \(| \Ups^{1}_{\ell_1} \cap \Ups^{2}_{\ell_2}|\), and the multiplicity of \(1\) is \(|\Ups^{1}_{\ell_1} \setminus \Ups^{2}_{\ell_2}| +  |\Ups^{2}_{\ell_2} \setminus \Ups^{1}_{\ell_1}|\). The proposition follows after observing that the right hand side of Eq.~\ref{eq: rel for sum of proj} equals \(\la(P_{\da \ell_1} (B_1^{\da} \otimes  C_1^{\da}) + P_{\da \ell_2} (B_2^{\da} \otimes C_2^{\da}))\).\qedhere
\end{proof}

We are now ready to prove Thm.~\ref{thm: sep Fan majorization}. 
\begin{proof} [Proof of Thm.~\ref{thm: sep Fan majorization}]
An immediate corollary of Prop.~\ref{prop: bound for alignment terms} is that 
\begin{gather}\cP_1^{(k)} (V_{\bullet} (B_1 \otimes C_1), V_{\bullet}(B_2 \otimes C_2)) \subseteq \cP_1^{(k)} (V_{\bullet} (B_1^{\da} \otimes C_1^{\da}), V_{\bullet} (B_2^{\da} \otimes C_2^{\da})).
\end{gather}
This implies that the corresponding linear programming bounds satisfy 
\begin{gather}
    u_k (B_1 \otimes C_1, B_2 \otimes C_2) \leq u_k (B_1^\da \otimes C_1^\da, B_2^\da \otimes C_2^\da).
\end{gather}
Then
\begin{align}
    s_k (B_1 \otimes C_1 + B_2 \otimes C_2) &\leq u_k (B_1 \otimes C_1, B_2 \otimes C_2) \\
    &\leq u_k (B_1^\da \otimes C_1^\da, B_2^\da \otimes C_2^\da) \\
    &= s_k (B_1^\da \otimes C_1^\da + B_2^\da \otimes C_2^\da),
\end{align}
where the equality follows from that fact that \(B_1^\da \otimes C_1^\da\) and \(B_2^\da \otimes C_2^\da\) are simultaneously diagonalizable and Thm.~\ref{thm: tightness for diagonal operators}.
\end{proof}

\subsection[Resolving the 2-letter spin alignment conjecture]{Resolving the \(2\)-letter spin alignment conjecture}
\label{subsec: spin align with 2 letters}
We use Eq.~\ref{eq: sep Fan majorization relations} to resolve the \(n=2\) case of the spin alignment conjecture. 
\begin{prop}
\label{prop: spin alignment n=2}
Let \(M\) denote a quantum state on \(\bC^d\) and \(p\) denote a probability distribution on the subsets of \(\{1,2\}\). Then for every pure quantum state tuple \((1, \psi_{\{1\}}, \psi_{\{2\}}, \psi_{\{1,2\}})\) and every unit vector \(v_M\) satisfying \(M v_M = \la_1(M) v_M\), it holds that
\begin{align}
\label{maj: n=2} 
\sum_{J \subseteq \{1,2\}} p_{J}\psi_{J} \otimes M^{\otimes J^c} \preceq \sum_{J \subseteq \{1,2\}} p_{J}({v_M {v_M}^*})^{\otimes J}  \otimes M^{\otimes J^c}.
\end{align} 
\end{prop}
\begin{proof}
Since the two operators \(p_{\varnot}  M \otimes M + p_{\{1\}} v_M v_M^* \otimes M + p_{\{2\}} M \otimes v_M v_M^*\) and \(p_{\{1,2\}} v_M v_M^* \otimes v_M v_M^*\) are perfectly aligned, it suffices to prove that \(p_{\varnot}  M \otimes M + p_{\{1\}} v_M v_M^* \otimes M + p_{\{2\}} M \otimes v_M v_M^*\) majorizes \(p_{\varnot}  M \otimes M + p_{\{1\}} \psi_1 \otimes M + p_{\{2\}} M \otimes \psi_2\). 

Combining the first two terms of the latter and using Eq.~\ref{eq: sep Fan majorization relations} yields
\begin{equation}
\label{eq:app of sep fan}
\begin{aligned}
( p_{\varnot}  M +p_{\{1\}} \psi_1)\otimes M &+ p_{\{2\}} M \otimes \psi_2 \preceq \\
&( p_{\varnot}  M +p_{\{1\}} \psi_1)^{\da} \otimes M^{\da} + p_{\{2\}} M^\da \otimes \psi_2^{\da}.
\end{aligned}
\end{equation}
The right hand side of this relation is unitarily equivalent to 
\begin{align}
\label{eq: half aligned in proof}
(p_{\varnot}  M +p_{\{1\}} \psi_1)^{\da} \otimes M +  p_{\{2\}} M^\da \otimes v_M v_M^*.
\end{align}
Define \(\tilde{M} := M - \la_1 (M) v_M v_M^*\). Expanding the second factor of the first term in Eq.~\ref{eq: half aligned in proof} yields
\begin{align}
(p_{\varnot}  M +p_{\{1\}} \psi_1)^{\da} \otimes \tilde{M} + (\la_1(M) (p_{\varnot}  M +p_{\{1\}} \psi_1)^{\da} +  p_{\{2\}} M^\da) \otimes v_M v_M^*.
\end{align}

Since the two terms are orthogonal, we can treat them separately (compare with~Prop III.8 in Ref.~\cite{Alhejji2024}). Fan's majorization relation Eq.~\ref{eq: Ky Fan majorization relation with da} and, for example, Theorem~2.2 in Ref.~\cite{Bondar2003} imply
\begin{align}
(p_{\varnot}  M +p_{\{1\}} \psi_1)^{\da} \otimes \tilde{M} \preceq (p_{\varnot}  M^\da +p_{\{1\}} \psi_1^\da) \otimes \tilde{M}.
\end{align}
The right hand side of this relation is unitarily equivalent to \( (p_{\varnot}  M +p_{\{1\}} v_M v_M^*) \otimes \tilde{M}\). By the same token, it holds that
\begin{equation}
\begin{aligned}
(\la_1(M) (p_{\varnot}  M &+p_{\{1\}} \psi_1)^{\da}  +  p_{\{2\}} M^\da) \otimes v_M v_M^* \\
&\preceq (\la_1(M) (p_{\varnot}  M^\da +p_{\{1\}} \psi_1^\da) +  p_{\{2\}} M^\da) \otimes v_M v_M^*.
\end{aligned}
\end{equation}
The right hand side of this relation is unitarily equivalent to 
\begin{align}
(\la_1(M) (p_{\varnot}  M +p_{\{1\}} v_M v_M^*) +  p_{\{2\}} M) \otimes v_M v_M^*. &\qedhere
\end{align}
\end{proof}

Therefore, by the reduction in Sec.~V of Ref.~\cite{Leditzky_2023}, the coherent information of any platypus channel is additive to the \(2\)-letter level.
\begin{cor}
\label{cor: platypus channel additivity}
If \(\cM\) is a platypus channel, then \(I^{\text{(coh)}}({\mathcal{M}}^{\otimes 2})=2 I^{\text{(coh)}} ({\mathcal{M}})\).
\end{cor}

\section{Concluding remarks}
\label{sec: conclusion}
The obstacle to generalizing Eq.~\ref{eq: sep Fan majorization relations} to three or more factors is that the proof of Prop.~\ref{prop: orthogonal tensor products} does not apply to cases with more than two factors.

The set of triples
\begin{gather}
\{(\la(A_1), \la(A_2), \la(A_1 + A_2)) \mid A_1, A_2 \; \text{self-adjoint on} \; \bC^{d_B} \otimes \bC^{d_C}\}    
\end{gather}
is a polyhedron that is completely described by Horn's inequalities \cite{Bhatia2001}. These include the inequalities given by Fan's majorization relation. The new majorization relation Eq.~\ref{eq: sep Fan majorization relations} gives new inequalities that account for the tensor product structure. Majorization relations involving tensor products feature in the theory of catalytic majorization \cite{Aubrun2008}. The relationship between the spectrum of a bipartite self-adjoint operator and the spectrum of one of its partial traces is explored in Ref.~\cite{DAFTUAR200580}. Connections between enumerative geometry, Horn's inequalities and the quantum marginal problem are described in Ref.~\cite{Knutson2009}.

Finally, here is an approach to proving Eq.~\ref{eq: sep Fan majorization relations} that does not work. Observe that the majorization relation
\begin{align}
    B_1^\da \otimes C_1 + B_2^\da \otimes C_2 \preceq  B_1^\da \otimes C_1^\da + B_2^\da \otimes C_2^\da
\end{align}
holds. Since majorization is transitive, Eq.~\ref{eq: sep Fan majorization relations} would be implied by the majorization relation
\begin{align}
\label{eq: wrong maj rel}
B_1 \otimes C_1 + B_2 \otimes C_2 \preceq^?  B_1^\da \otimes C_1 + B_2^\da \otimes C_2.
\end{align}
However, this latter majorization relation does not hold in general. 
\begin{example}
\label{ex: counter example to one da}
Consider the case where \(d_B = d_C = 2\), and denote the superposition \(e := \frac{1}{\sqrt{2}} (q_1 + q_2)\). Let \(\tilde{B}_1 := e e^* =:\tilde{C}_1\), \(\tilde{B}_2 := 2 q_1 q_1^* + q_2 q_2^*\) and \(\tilde{C}_2 := q_1 q_1^*\). Then
\begin{gather}
\tilde{B}_1 \otimes \tilde{C_1} + \tilde{B}_2 \otimes \tilde{C}_2 = e e^* \otimes e e^* + 2 q_1 q_1^* \otimes q_1 q_1^* + q_2 q_2^* \otimes q_1 q_1^*
\end{gather}
and 
\begin{gather}
\tilde{B}^\da_1 \otimes \tilde{C_1} + \tilde{B}^\da_2 \otimes \tilde{C}_2 = q_1 q_1^* \otimes e e^* +  2 q_1 q_1^* \otimes q_1 q_1^* + q_2 q_2^* \otimes q_1 q_1^*.
\end{gather}
It can be verified that 
\begin{gather}
s_2 (\tilde{B}^\da_1 \otimes \tilde{C_1} + \tilde{B}^\da_2 \otimes \tilde{C}_2) - s_2(\tilde{B}_1 \otimes \tilde{C_1} + \tilde{B}_2 \otimes \tilde{C}_2) < -0.05.    
\end{gather}
\end{example}
The relation Eq.~\ref{eq: wrong maj rel} holds in cases where \(B_1\) and \(B_2\) are both rank-\(1\). This follows from a majorization relation for separable quantum states due to Nielsen and Kempe \cite{Nielsen2001}. 

\subsubsection*{Acknowledgments}
I thank Manny Knill for many hours of stimulating and encouraging discussions, and Milad Marvian for pushing me to find counterexamples. I am grateful to Felix Leditzky for comments that improved the presentation of the material. This work was supported by National Science Foundation Grant PHY-2116246.

%\subsection*{Declarations}
%\textbf{Conflicts}: I declare no conflict of interests. \textbf{Data availability}: the manuscript has no associated data.

\appendix
\section{Remainder of the proof of Thm.~\ref{thm: tightness for diagonal operators}}
\label{sec: appendix}
Following the strategy outlined in Sec.~\ref{subsec: tightness}, we show that for every feasible pair, there exists a symmetric feasible pair that does at least as well. It is useful to observe that for every \((S_1,S_2)\) such that \(S_1, S_2 \subseteq [d]\) and \(|S_1|=|S_2|=k\), in addition to the basic constraints, certain alignment constraints are satisfied as well. 
\begin{lem}
\label{lem: some align are ok}
For a \(k\)-set pair \((S_1, S_2)\) and \((\ell_1, \ell_2) \in [d] \times [d]\), 
\begin{align}
|S_1 \cap \Og^1_{\ell_1}| + |S_2 \cap \Og^2_{\ell_2}| \leq \min (2 k, |\Og^{1}_{\ell_1} \cap \Og^{2}_{\ell_2}| + |\Og^{1}_{\ell_1} \cup  \Og^{2}_{\ell_2}|).
\end{align}
\end{lem}
\begin{proof}
The fact that \(|S_1|=|S_2|=k\) implies \(|\Og^{1}_{\ell_1} \cap S_1| + |\Og^{2}_{\ell_2} \cap S_2| \leq 2 k\). The other inequality follows from inclusion-exclusion:
\[ |\Og^{1}_{\ell_1} \cap S_1| + |\Og^{2}_{\ell_2} \cap S_2| \leq |\Og^{1}_{\ell_1}| + |\Og^{2}_{\ell_2}| = |\Og^{1}_{\ell_1} \cap \Og^{2}_{\ell_2}| + |\Og^{1}_{\ell_1} \cup  \Og^{2}_{\ell_2}|.\qedhere \] \end{proof}
Therefore, since we are optimizing over \(k\)-set pairs \((S_1, S_2)\), we need only be concerned with constraints of the form
\begin{align}
\label{eq: comb align cons reduced}
|S_1 \cap \Og^1_{\ell_1}| + |S_2 \cap \Og^2_{\ell_2}| \leq | \Og^{1}_{\ell_1} \cap \Og^{2}_{\ell_2} | + k.
\end{align}

We start by showing that every symmetric pair of \(k\)-sets is feasible. This shows explicitly that the feasible set is not empty. 
\begin{lem}
\label{lem: every sym pair is feasible}
If \(S \subseteq [d]\) satisfies \(|S| = k\), then \((S, S)\) is a feasible pair. 
\end{lem}
\begin{proof}
Since \(|S| = k\), the vector \(\one_S \oplus \one_S\) satisfies the basic constraints. For \((\ell_1, \ell_2) \in [d] \times [d]\), using inclusion-exclusion and the fact that intersections distribute over unions yields
\begin{align*}
|S \cap \Og^1_{\ell_1}| + |S \cap \Og^2_{\ell_2}| &= | (S \cap \Og^1_{\ell_1}) \cap (S \cap \Og^2_{\ell_2})| + | (S \cap \Og^1_{\ell_1}) \cup (S \cap \Og^2_{\ell_2})| \\
&= | S \cap (\Og^1_{\ell_1} \cap \Og^2_{\ell_2})| + | S \cap (\Og^1_{\ell_1} \cup \Og^2_{\ell_2})| \\
&\leq  \min(k, | \Og^{1}_{\ell_1} \cap \Og^{2}_{\ell_2} |) + \min(k, |\Og^{1}_{\ell_1} \cup \Og^{2}_{\ell_2}|). \qedhere \end{align*} \end{proof}

Let \((S_1, S_2)\) be an arbitrary feasible pair. We argue into two steps that it may be adjusted to a symmetric pair without decreasing the objective value. In the first step, we fix \(S_1\) and optimize the second argument. Specifically, we find the optimal \(k\)-set \(S \subseteq [d]\) such that \((S_1, S)\) is a feasible pair. In the second step, we show that the objective value at \((S, S)\) is greater than or equal to the objective value at \((S_1, S)\).

\subsection{First step}
\label{subsubsec: first}

Denote the elements of \(S_1\) with \( x_1, \ldots, x_k\) such that \(x_k \geq^1 \cdots \geq^1 x_2 \geq^1 x_1\). For each \(i \in [k]\), let \(m_{i} \in [d]\) be minimal such that \(x_i \in \Og^1_{m_i}\). That means \(m_i\) is the order of \(x_i\) according to \(\geq^1\). This implies \(\Og^1_{m_1} \subset \Og^1_{m_2} \subset \cdots \subset \Og^1_{m_k}\). 

Let \(i' \in [k]\) be given. Keeping in mind the constraints in Eq.~\ref{eq: comb align cons reduced}, we consider the set of all \(r \in [d]\) such that for all \(i \in [k]\), the constraint 
\begin{align}
\label{eq: constraint for first step}
i + i' \leq k + | \Og^1_{m_i} \cap \Og^2_{r}|
\end{align}
is satisfied. This set is not empty because for all \(i \in [k]\), the intersection \(\Og^1_{m_i} \cap \Og^2_{d} = {\Og^1_{m_i}} \cap {[d]} = \Og^1_{m_i}\), and so \(k + | \Og^1_{m_i} \cap \Og^2_{d}| = k + m_i \geq i' + i\). Define \(r_{i'}\) to be minimal in this set.

\begin{lem}
\label{lem: minimality implies equalities}
For each \(i' \in [k]\), there exists some \(i \in [k]\) such that 
\begin{gather}
i + i' = k + | \Og^1_{m_i} \cap \Og^2_{r_{i'}}|.
\end{gather}
\end{lem}
\begin{proof}
The minimality of \(r_{i'}\) implies that either \(r_{i'} = 1\) or \(r_{i'} > 1\) and there exists \(i \in [k]\) such that \(i + i' > k + | \Og^1_{m_i} \cap \Og^2_{r_{i'}-1}|\). 

If \(r_{i'} = 1\), then the constraint in Eq.~\ref{eq: constraint for first step} corresponding to \(k\) reads 
\begin{gather}
k + i' \leq k + \underbrace{| \Og^1_{m_k} \cap \Og^2_{1}|}_{\leq 1}, 
\end{gather}
implying that \(i' = 1\) and \(| \Og^1_{m_k} \cap \Og^2_{1}| = 1\). 

If \(r_{i'} > 1\), then since \(| \Og^1_{m_i} \cap \Og^2_{r_{i'}}| - | \Og^1_{m_i} \cap \Og^2_{r_{i'}-1}| \leq 1\), it must hold that \(i + i' = k + | \Og^1_{m_i} \cap \Og^2_{r_{i'}}|\). \end{proof}

For each \(i' \in [k]\), we denote the maximal element in \(\Og_{r_{i'}}^2\) according to the \(\geq^2\) order by \(y_{i'}\). The set we seek is \(S:=\{ y_1, \ldots, y_k\}\). 
\begin{lem}
\label{lem: y's are distinct and monotonic}
The set \(S\) is a \(k\)-set and its elements satisfy \begin{gather}
y_k \geq^2 \cdots \geq^2 y_2 \geq^2 y_1.
\end{gather}
\end{lem}
\begin{proof}
Let \(i'_1, i'_2 \in [k]\) be such that \(i'_1 < i'_2\). By Lem.~\ref{lem: minimality implies equalities}, there exists at least one \(i \in [k]\) such that \(i + i'_1 = k + | \Og^1_{m_i} \cap \Og^2_{r_{i'_1}}|\). Then Eq.~\ref{eq: constraint for first step} with \(i' = i'_2\) implies
\begin{gather}
k + | \Og^1_{m_i} \cap \Og^2_{r_{i'_1}}| =  i + i'_1 < i + i'_2 \leq k+| \Og^1_{m_i} \cap \Og^2_{r_{i'_2}}|.  
\end{gather} 
Therefore, \(r_{i'_2}\) has to be strictly larger than \(r_{i'_1}\). This shows that \(|S| = k\) and \(y_k \geq^2 \cdots \geq^2 y_2 \geq^2 y_1\).
\end{proof}

\begin{lem}
\label{lem: feasiblity of S}
The pair \((S_1, S)\) is a feasible pair. 
\end{lem}
\begin{proof}
By Lem.~\ref{lem: y's are distinct and monotonic}, \(|S|=k\) and so \((S_1, S)\) is a \(k\)-set pair. Therefore, by Lem.~\ref{lem: some align are ok}, it suffices to check that for all \((\ell_1, \ell_2) \in [d] \times [d]\), 
\begin{align}
|S_1 \cap \Og^1_{\ell_1}| + |S \cap \Og^2_{\ell_2}| \leq | \Og^{1}_{\ell_1} \cap \Og^{2}_{\ell_2} | + k.
\end{align}
If \(S_1 \cap \Og^1_{\ell_1}\) or \(S \cap \Og^1_{\ell_2}\) is empty, then the constraint is satisfied. Otherwise, let \(i \in [k]\) be maximal such that \(m_{i} \leq \ell_1\), and \(i' \in [k]\) be maximal such that \(r_{i'} \leq \ell_2\). Then 
\begin{align}
|S_1 \cap \Og^1_{\ell_1}| + |S \cap \Og^2_{\ell_2}| &= |S_1 \cap \Og^1_{m_i}| + |S \cap \Og^2_{r_{i'}}| \\
&= i + i'\\
&\leq k + | \Og^1_{m_i} \cap \Og^2_{r_{i'}}| \\
&\leq k + | \Og^{1}_{\ell_1} \cap \Og^{2}_{\ell_2} |,
\end{align} 
where the last inequality follows from the containments \(\Og^1_{m_i} \subseteq \Og^1_{\ell_1}\) and \(\Og^2_{r_{i'}} \subseteq \Og^2_{\ell_2}\).
\end{proof}

Lastly, we show that the pair \((S_1, S)\) is at least as good as \((S_1, S_2)\).
\begin{lem}
\label{lem: optimal S}
The objective value at \((S_1, S)\) is greater than or equal to the objective value at \((S_1, S_2)\).
\end{lem}
\begin{proof}
Let \(S_2 = \{ z_1, \ldots, z_k\}\), denoted so that \(z_k \geq^2 \cdots \geq^2 z_2 \geq^2 z_1\). Recall that for each \(i' \in [k]\), \(r_{i'}\) is the order of the eigenvalue \(\tilde{\la}_{y_{i'}} (D_2)\). Since \(r_{i'}\) is chosen to be minimal, the inequality \(z_{i'} \geq^2 y_{i'}\) holds. This implies 
\begin{align}
    \tilde{\la}_{y_{i'}}(D_2) - \tilde{\la}_{z_{i'}}(D_2) \geq 0.
\end{align}
The difference between objective value at \((S_1, S)\) and the objective value at \((S_1, S_2)\) is  \(\tilde{\la} (D_2)^T (\one_{S} - \one_{S_2})\) which equals 
\(\sum_{i'=1}^k \tilde{\la}_{y_{i'}}(D_2) - \tilde{\la}_{z_{i'}}(D_2) \geq 0 \).
\end{proof}

\subsection{Second step}
\label{subsubsec: second}
Now we show that that the objective value at \((S, S)\) is greater than or equal to the objective value at \((S_1, S)\). We do so by showing that for each \(i \in [k]\), there exist at least \(i\) elements in \(S\) that precede \(x_i\) in the \(\geq^1\) order.

\begin{lem}
\label{lem: the other foot}
For each \(i \in [k]\), \(| S \cap \Og^1_{m_i}| \geq i\).    
\end{lem}
\begin{proof}
We argue by strong induction on \(p := k - i +1\). We start with the base case \(p=1\). As in the proof of Lem.~\ref{lem: minimality implies equalities}, for each \(i'\), there exists \(i \in [k]\) such that 
\begin{gather}
 i + i' = k + | \Og^1_{m_i} \cap \Og^2_{r_{i'}}|   
\end{gather}
and either \(r_{i'} = 1\), or \(r_{i'} > 1\) and \(| \Og^1_{m_i} \cap \Og^2_{r_{i'}}| -  | \Og^1_{m_i} \cap \Og^2_{r_{i'}-1}| = 1\). If \(r_{i'} = 1\), then \(i'=1\) and \(y_{1} \in \Og^1_{m_k}\). If \(r_{i'} > 1\), then \(y_{i'} \in \Og^1_{m_{i}} \subseteq \Og^1_{m_k}\). Therefore, \(S \subseteq \Og^1_{m_k}\), and so \(| S \cap \Og^1_{m_k}| = k\).

Define \(\ga(1) := k\). For \(i' \in [k]\) such that \(i' > 1\), define \(\ga(i') \in [k]\) to be minimal such that 
\begin{gather}
\ga(i') + i' = k + | \Og^1_{m_{\ga(i')}} \cap \Og^2_{r_{i'}}| > k + | \Og^1_{m_{\ga(i')}} \cap \Og^2_{r_{i'}-1}|.
\end{gather}

Suppose that \(p > 1\) and the statement is true for all elements of \([p-1]\). Let \(T = \{ i'_1, \ldots, i'_{|T|}\}\) be the subset of indices \(i' \in [k]\) that satisfy \(k-\ga(i')+1 < p\), denoted so \(i'_1 < \ldots < i'_{|T|}\). Notice that \(i'_1 = 1\). Observe that for \(i' \in T^c\), \(y_{i'} \in \Og^1_{m_{\ga(i')}}\) implies   \(y_{i'} \in \Og^1_{m_{k-p+1}}\). 

For each \(i' \in T\), we have the equality
\begin{align}
\label{eq: exc inc with o and i}
| \Og^1_{m_{\ga(i')}} \cap \Og^2_{r_{i'}}| = \ga(i') + i' - k
\end{align}
and \(k-\ga(i')+1 < p\). By the inductive hypothesis,  \(|\Og^1_{m_{\ga(i')}} \cap S| \geq \ga(i')\). 
Because of the containment \((\Og^1_{m_{\ga(i')}} \cap S) \cap (\Og^2_{r_{i'}} \cap S) \subseteq \Og^1_{m_{\ga(i')}} \cap \Og^2_{r_{i'}}\) and the inclusion-exclusion inequality
\begin{gather}
| (\Og^1_{m_{\ga(i')}} \cap S) \cap (\Og^2_{r_{i'}} \cap S)| \geq \ga(i') + i' - k
\end{gather}
the equality Eq.~\ref{eq: exc inc with o and i} implies 
\begin{gather}
\label{eq: equality of sets}
(\Og^1_{m_{\ga(i')}} \cap S) \cap (\Og^2_{r_{i'}} \cap S) = \Og^1_{m_{\ga(i')}} \cap \Og^2_{r_{i'}},
\end{gather}
which implies \(\Og^1_{m_{\ga(i')}} \cap \Og^2_{r_{i'}} \subseteq S\). This in turn implies \(\Og^1_{m_{k-p+1}} \cap \Og^2_{r_{i'}} \subseteq S\) for all \(i' \in T\). The constraint Eq.~\ref{eq: comb align cons reduced} implies
\begin{gather} 
|\Og^1_{m_{k-p+1}} \cap \Og^2_{r_{i'_{|T|}}}| \geq k-p+1+ i'_{|T|} - k = i'_{|T|} - p + 1.
\end{gather}
Since \(({\Og^1_{m_{k-p+1}}} \cap {\Og^2_{r_{i'_{|T|}}}}) \subseteq ({\Og^2_{r_{i'_{|T|}}}} \cap S)\), by inclusion-exclusion
\begin{align}
    |({\Og^1_{m_{k-p+1}}} \cap {\Og^2_{r_{i'_{|T|}}}}) \cap \{y_{i'_1}, y_{i'_2}, \ldots, y_{i'_{|T|}}\}| &\geq i'_{|T|} -p + 1 + |T| - i'_{|T|} \\
    &= |T| - p +1
\end{align}
Together with the elements whose subscripts are in \(T^c\), this implies
\( | S \cap \Og^1_{m_{k-p+1}}| \geq k - |T| + |T| - p + 1 = k - p + 1.\) \end{proof}

\begin{lem}
\label{lem: final lemma}
 The objective value at \((S, S)\) is greater than or equal to the
objective value at \((S_1, S)\). 
\end{lem}
\begin{proof}
We relabel the elements of \(S\) to reflect their order according to \(\geq^1\). Let us denote them with \(x'_1, x'_2, \ldots, x'_k\) such that \(x'_k \geq^1 \cdots \geq^1 x'_2 \geq^1 x'_1\). 

By Lem.~\ref{lem: the other foot}, for each \(i \in [k]\), there exist at least \(i\) elements in \(S\) that are in \(\Og^1_{m_{i}}\). In particular, this implies that \(\{x'_1, \ldots, x'_i\} \subseteq \Og^1_{m_{i}}\). Hence, for all \(i \in [k]\), \(x_i \geq^1 x'_i\).    

The difference between objective value at \((S, S)\) and the objective value at \((S_1, S)\) is  \(\tilde{\la}(D_1)^T (\one_{S} -  \one_{S_1})\) which equals
\(\sum_{i=1}^k \tilde{\la}_{x'_i}(D_1) - \tilde{\la}_{x_i} (D_1) \geq 0\). \end{proof}

\bibliographystyle{plainurl}
\bibliography{references}
\end{document}